\documentclass[10pt,journal,compsoc]{IEEEtran}
\ifCLASSOPTIONcompsoc
  \usepackage[nocompress]{cite}
\else
  \usepackage{cite}
\fi
\ifCLASSINFOpdf
\else
\fi

\usepackage{amssymb,amsmath,amsthm}

\usepackage{enumitem}

\newtheorem{Corollary}{Corollary}

\newtheorem{theorem}{Theorem}

\usepackage{empheq}
\usepackage{cite}
\usepackage{amsmath,amssymb,amsfonts}
\usepackage{algorithmicx}
\usepackage{graphicx}
\usepackage{textcomp}
\usepackage{xcolor}
\usepackage{subcaption} 
\usepackage{pgf}
\usepackage{lmodern}
\usepackage{moeptikz}

\usepackage{array}
\usepackage{tabularx}
\newcolumntype{C}{>{\centering\arraybackslash}X} 
\usepackage{pifont}
\usepackage{booktabs}
\usepackage{adjustbox}
\usepackage{graphicx}
\usepackage{multirow}
\usepackage{algorithm}
\usepackage{algpseudocode}
\usepackage{etoolbox}
\usepackage{makecell}
\usepackage{url}
\usepackage{comment}
\usepackage{tikz}
\usetikzlibrary{positioning, arrows, shapes}
\usepackage{pgfplots}
\usepackage{amsmath}
\usepackage{amssymb}
\usepackage{amsthm}
\pgfplotsset{compat=newest}
\pgfplotsset{plot coordinates/math parser=false}
\usepackage{makecell}
\makeatletter
\renewcommand{\ALG@beginalgorithmic}{\small}
\patchcmd{\Statex}{\hskip\algorithmicindent}{}{}{}
\makeatother
\newcommand{\cmark}{\ding{51}}%
\newcommand{\xmark}{\ding{55}}%
\def\BibTeX{{\rm B\kern-.05em{\sc i\kern-.025em b}\kern-.08em
    T\kern-.1667em\lower.7ex\hbox{E}\kern-.125emX}}

\begin{document}
\title{A Cost-efficient Credit-Based Shaper Deployment Framework for Time-Sensitive Networks}
%
%
%
%

\author{Santiago~Torres-Borda and Ahlem~Mifdaoui
\IEEEcompsocitemizethanks{\IEEEcompsocthanksitem ISAE SUPAERO, Université de Toulouse, France.\protect\\
E-mail: $\{$firstname.lastname$\}$@isae-supaero.fr}
\thanks{This article has been accepted for publication in the IEEE International Symposium On Real-Time Distributed Computing Proceedings. This is the author’s version
which has not been fully edited and content may change prior to final publication.}}

%
%


\IEEEtitleabstractindextext{%

\begin{abstract}
Time-sensitive networks are designed to meet stringent Quality of Service (QoS) requirements for mixed-criticality traffic with diverse performance demands. Ensuring deterministic guarantees for such traffic while reducing deployment costs remains a significant challenge. This paper proposes a cost-efficient partial deployment strategy for Time Sensitive Networking (TSN) devices within legacy Ethernet network. At the core of our approach is the Credit-Based Shaper (CBS), a key TSN scheduling mechanism. Unlike cost-prohibitive full CBS deployment, our approach selectively integrates CBS where it is most needed to enhance performance while reducing costs. Combining Network Calculus for schedulability verification and a heuristic optimization method for CBS configuration and placement, our proposal minimizes deployment costs while improving schedulability for medium-priority traffic and mitigating blocking delays for high-priority traffic. The feasibility and benefits of our approach are validated on a realistic automotive TSN use case with up to 70\% of reduction in TSN devices requirements compared to a full deployment.
\end{abstract}

}

\maketitle

\IEEEdisplaynontitleabstractindextext

\section{Introduction}
The rapid growth of time-sensitive networks across industries such as automotive systems, industrial automation, and aerospace highlights their essential role in enabling deterministic communication. These networks are designed to meet stringent Quality of Service (QoS) requirements of mixed-criticality traffic with diverse performance demands. However, ensuring deterministic guarantees for heterogeneous traffic remains a challenging issue, particularly when there is a balance need between network performance and deployment costs.

A widely used solution in this area is based on high-rate Ethernet technology supporting QoS for up to eight traffic classes, IEEE 802.1Q \cite{8021QStd}, providing low costs and easy deployment. However, its reliance on a basic Non-Preemptive Strict-Priority (NP-SP) scheduling, known to be unfair, can result in schedulability issues and even starvation for medium- and low-priority traffic. More recently, Time-Sensitive Networking (TSN) standards \cite{TimeSensitiveNetworking} extending IEEE 802.1Q with advanced scheduling and reliability mechanisms are seen as a promising solution to improve determinism and availability. Despite their benefits, these standards are often cost-prohibitive, with TSN switch prices ranging from \texteuro 2,000 to \texteuro 10,000, in addition to being immature for industrial-scale adoption.

To address this issue, we propose a novel partial deployment strategy that integrates TSN devices selectively within a legacy Ethernet network. This approach capitalizes on the interoperability of TSN and Ethernet technologies to improve QoS and schedulability where it is most needed, without the high cost of full TSN deployment. By strategically deploying TSN components within an Ethernet network, our approach incurs a beneficial balance between network performance and cost-efficiency.

At the core of our approach is the Credit-Based Shaper (CBS), described in IEEE 802.1Qav \cite{Qav}, a key TSN scheduling mechanism. CBS enhances fairness by shaping high-priority traffic to reduce its burstiness, thereby decreasing scheduling delays for medium- and low-priority traffic. While prior research \cite{diemer_formal_2012,ashjaei_schedulability_2017,cao_independent_2018,zhao_minimum_2024} assumes full CBS deployment, our study demonstrates the benefits of a partial CBS deployment approach, effectively mitigating schedulability issues for lower-priority traffic and blocking effects for shaped higher-priority traffic \cite{diemer_formal_2012}, all while reducing costs. Our primary contributions in this paper are:
\begin{itemize}
    \item A Cost-efficient CBS deployment framework based on Network Calculus\cite{le_boudec_network_2001} theory to verify schedulability and widely used for TSN applications\cite{zhao_improving_2019,zhao_minimum_2024,mohammadpour_end--end_2018,Thomas,de_azua_complete_2014,torres-borda_towards_2024}. This theory is combined with a heuristic optimization approach for CBS configuration and placement within a minimum number of switches to reduce the deployment costs;
    \item The formulation of the partial CBS deployment as an optimization problem and its solving based on a cost-aware heuristic algorithm alongside its complexity analysis; 
    \item Validation of the benefits of such an approach on a representative automotive TSN use case.
\end{itemize}

The rest of this paper is organized as follows. Section II reviews related work. Section III presents the main system assumptions and model. Sections IV to VII detail our proposed framework overview, the schedulability verification method, the partial CBS deployment formulation and algorithm, respectively. Afterwards, we evaluate the performance of our proposal and validate its benefits on a realistic automotive TSN use case with reference to existing solutions in Section VIII. Finally, Section IX concludes this paper.

\section{Related Work} 
\label{sec: related work}
The most relevant existing works concerning CBS deployment are depicted in Table \ref{tab:studies overview}. The main supported assumptions that we consider to categorize these approaches are: (i) the placement paradigm to highlight the potential support of partial CBS placement; (ii) the legacy support to enable the use of legacy Ethernet devices, e.g., switches (SW) or End-systems (ES)); (iii) the configuration strategy guaranteeing the schedulability of shaped traffic and/or non-shaped traffic.

\begin{table}[b]
    \centering
    \begin{tabular}{>{\centering\arraybackslash}m{0.15\linewidth}|>{\centering\arraybackslash}m{0.05\linewidth}>{\centering\arraybackslash}m{0.1\linewidth}|>{\centering\arraybackslash}m{0.1\linewidth}>{\centering\arraybackslash}m{0.15\linewidth}|>{\centering\arraybackslash}m{0.05\linewidth}>{\centering\arraybackslash}m{0.05\linewidth}} 
         &  \multicolumn{2}{c|}{\textbf{Placement}} & \multicolumn{2}{c|}{\textbf{Configuration}} & \multicolumn{2}{c}{\textbf{Legacy}} \\ 
         & \multicolumn{2}{c|}{\textbf{Paradigm}} & \multicolumn{2}{c|}{\textbf{Strategy}} & \multicolumn{2}{c}{\textbf{Support}} \\ \cline{1-7}
         \multicolumn{1}{c|}{\textit{Study}} & \textit{Full CBS} & \textit{Partial CBS} & \textit{Ensure schedulability for shaped flows} & \textit{Ensure schedulability for non-shaped flows} & \textit{SW} & \textit{ES} \\ \hline \hline
         Diemer et al. \cite{diemer_formal_2012} & \cmark & \xmark & \cmark & \cmark & \xmark & \xmark \\ \hline 
         Ashjaei et al. \cite{ashjaei_schedulability_2017} & \cmark & \xmark & \cmark & \xmark & \xmark & \xmark \\ \hline 
         Cao et al. \cite{cao_independent_2018} & \cmark & \xmark & \cmark & \xmark & \xmark & \xmark \\ \hline 
         Maile et al. \cite{maile_delay-guaranteeing_2022} & \xmark & \cmark & \cmark & \xmark & \xmark & \cmark \\ \hline 
         Zhao et al. \cite{zhao_minimum_2024} & \cmark & \xmark & \cmark & \xmark & \xmark & \xmark \\ \hline 
         This work & \cmark & \cmark & \cmark & \cmark & \cmark & \cmark \\ \hline
    \end{tabular}
     \caption{Related Works Overview}
     \label{tab:studies overview}
\end{table}

Most of these works consider \textit{Full CBS deployment} within the network and no \textit{legacy support} of Ethernet Switches (SW) and End Systems (ES), except in \cite{maile_delay-guaranteeing_2022} where authors consider the use of legacy ES not implementing CBS mechanism. These  assumptions are accounted in this present work to enable the use of legacy devices and of CBS only when it is needed to significantly decrease the deployment costs. Moreover, most works tend to focus on finding the CBS configuration that enables schedulability of shaped classes without considering the impact on non-shaped classes, except in \cite{diemer_formal_2012} that considers schedulability constraints of both classes. This shows a clear void in the literature that we aim to fill in our proposal.

The primary concern with CBS configuration is efficiently reserving bandwidth to ensure flow schedulability. According to IEEE Std 802.1Qav \cite{Qav}, the IdleSlope should be set as the bandwidth requirement of the class subject to CBS.  

Diemer et al. \cite{diemer_formal_2012} argue that this method of bandwidth reservation may lead to unscheduled shaped flows. They propose an over-reservation approach, which involves multiplying the IdleSlope specified in \cite{Qat} by an arbitrary factor. They also highlight the potential blocking effect of CBS, which can significantly increase latency for prioritized traffic, emphasizing the crucial role of IdleSlope selection. Ashjaei et al. \cite{ashjaei_schedulability_2017} use an analytical approach based on busy period analysis to compute an IdleSlope that ensures the schedulability of shaped traffic, integrated with scheduled traffic mechanisms from \cite{8021QStd} Cao et al. \cite{cao_independent_2018} propose an approach to compute bandwidth reservation which is independent from interference caused by other traffic classes. This method has shown to be more precise for single-switch scenarios. Maile et al. \cite{maile_delay-guaranteeing_2022} and Zhao et al. \cite{zhao_minimum_2024} utilize CBS Network Calculus \cite{le_boudec_network_2001} models from \cite{de_azua_complete_2014, mohammadpour_end--end_2018, zhao_improving_2019} to derive IdleSlopes for CBS deployed in complex networks with multiple switches implementing CBS. Both approaches use a method that derives the necessary Service Curve to schedule network flows and, from this Service Curve, compute an appropriate IdleSlope for configuring CBS. Maile et al. adjust flow paths if resources are insufficient for reservation, while Zhao et al. compute local delays using the Service Curves per output port.  Nitta et al. \cite{nitta_proposal_2022} propose a solution based on simulation models, where they introduce a constraint-based linear optimization model to compute suitable IdleSlope values based on end-to-end delays for shaped flows.

In our previous work \cite{torres-borda_towards_2024}, we detailed an illustrative example that highlights the blocking effect caused by CBS, as initially discussed in \cite{diemer_formal_2012}, and how using CBS selectively on specific output ports could mitigate this blocking effect. This focus on partial CBS deployment broadened the network's capabilities without requiring full CBS support across all switches. This paper is a continuation of our previous work providing a detailed cost-efficient CBS deployment framework and the analyses of its benefits on an automotive use case.

\section{System Model}
A time-sensitive network consists of end-systems (ES) and switches (SW), also called devices $dev\in DEV$, that are interconnected via Full-Duplex links with a transmission capacity $C$. An example of such network is illustrated in Fig \ref{fig:RepresentativeAvionicNetworkArchitecture}. Our network model is based on the underlying directed graph $Net=(OP, E)$, where $OP$ represents the set of output ports and $E$ represents the set of unidirectional links between a source and destination output ports. 

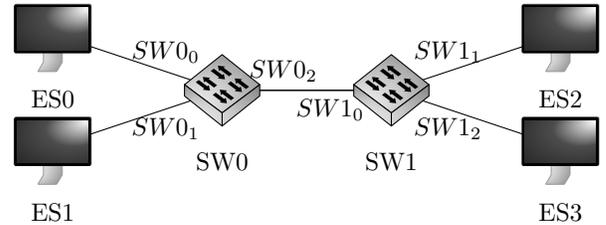
\begin{figure}[h!]
    \centering
    \begin{tikzpicture}[scale=0.75]
        \node[client] (client0) at (-5,0) {ES0};
        \node[client] (client1) at (-5,-2) {ES1};
        
        \node[switch] (switch0) at (-2,-1) {SW0};
        \node[switch] (switch1) at (1,-1) {SW1};
        
        \node[client] (client2) at (4,0) {ES2};
        \node[client] (client3) at (4,-2) {ES3};
        
        \draw (client0) -- (switch0) node[near end, above, xshift=0pt] {$SW0_0$};
        \draw (client1) -- (switch0) node[near end, below, xshift=0pt] {$SW0_1$};
        \draw (switch0) -- (switch1) node[near start, above, xshift=0pt] {$SW0_2$};
        \draw (switch0) -- (switch1) node[near end, below, xshift=0pt] {$SW1_0$};
        \draw (switch1) -- (client2) node[near start, above, xshift=0pt] {$SW1_1$};
        \draw (switch1) -- (client3) node[near start, below, xshift=0pt] {$SW1_2$};
        
    \end{tikzpicture}
    \vspace{10pt}
    \caption{Example of Time-sensitive Network Architecture}
    \label{fig:RepresentativeAvionicNetworkArchitecture}
\end{figure}

Each output port in $OP$ supports a NP-SP scheduler and eight priorities queues $P=\{0,1,2,3,4,5,6,7\}$, where $p=0$ represents the highest priority and $p=7$ represents the lowest one. Moreover, each output port has a device attribute $dev\in DEV$, denoted as $op.dev$. For each priority queue in $op$, we can enable (if needed) the CBS mechanism on top of NP-SP. To this effect, we use $CBS_{op}^p$ a binary variable equal to $1$ if there is CBS within $op$ for priority $p$. Each CBS-shaped queue with priority $p$ within output port $op$ has a reserved bandwidth called IdleSlope, denoted as $I_{op}^p$. The use of CBS within end-systems is optional in our approach, as we consider them as legacy devices that do not support complex queuing algorithms. Nevertheless, the approach presented in this paper is still valid if they are considered for CBS deployment.

For each flow $f \in \mathcal{F}$ traversing the network, we define the tuple $(r^f, b^f, L^f, D^f, p^f, \phi^f)$ for the flow's rate, burst, maximum frame size, deadline, assigned priority, and path, respectively. Specifically, for each flow $f$, its path $\phi^f$ is an ordered collection of crossed output ports from its source until its destination. The paths are known a priori and may have cyclic dependencies. We call $\mathcal{F}(p)$ the set of flows with priority $p$, $\mathcal{F}(op,p)$ those additionally crossing the output port $op$, and $\mathcal{F}(dev)$ those traversing device $dev$.

The main notations used in this paper are in Table \ref{tab:variable_definitions}, where upper indices indicate flows or a traffic priority and lower indices indicate a node or a set of nodes.
\begin{table}[h!]
    \caption{Notations}
   \label{tab:variable_definitions}
    \centering
   \small
   \begin{tabular}{|c|m{0.7\columnwidth}|}
       \hline
       $OP$ & The set of output ports in the network \\ 
       $E$ & The set of unidirectional links in the network \\ 
      $op$ & An output port within the set $OP$. \\ 
       $C$ & The transmission capacity of a link \\ 
       $op.dev$ & The device to which $op$ is associated \\ 
     $DEV$ & The set of devices, including switches ($SW$) and End systems ($ES$) \\ 
     $OP(dev)$ & the set of output ports of device $dev$ \\
       $r^f$ & The rate of flow $f$ \\
     $b^f$ & The burst size of flow $f$ \\
        $L^{f}$ & The maximum frame size of flow $f$ \\
       $D^f$ & The deadline of flow $f$ \\
      $p^f$ & The assigned priority of flow $f$. \\
      $\phi^f$ & The path of flow $f$ \\
       $L^{>p,max}$ & The maximum frame size among flows with priority $> p$ \\
       $I_{op}^p$ & The IdleSlope for a CBS at output port $op$ for priority $p$ \\
       $CBS_{op}^p$ & Binary variable indicating the presence of a CBS at $op$ for priority $p$ \\
      $\mathcal{F}(p)$ & the set of flows with priority $p$ \\
      $\mathcal{F}(op,p)$ & the set of flows with priority $p$ crossing the output port $op$ \\ 
      $\mathcal{F}(dev)$ & the set of flows traversing device $dev$ \\\hline
   \end{tabular}
\end{table}

\section{Framework Overview}
\label{sec:overview}
Our cost-efficient CBS deployment framework, illustrated in Fig.\ref{fig:framework_overview}, considers as input a time-sensitive network implementing only legacy NP-SP schedulers within all the output ports and having some schedulability issues that do not concern highest-priority flows. This particular assumption takes into account industrial feedback on the interest of using TSN standards only when the legacy IEEE 802.1Q standard is not sufficient to achieve the QoS requirements, due to the high deployment costs of TSN devices. The aim of our proposed framework is to generate an enriched network configuration with selectively placed CBS to improve schedulability while reducing costs.

\begin{figure}[h!]
\centering
\includegraphics[width=0.45\textwidth]{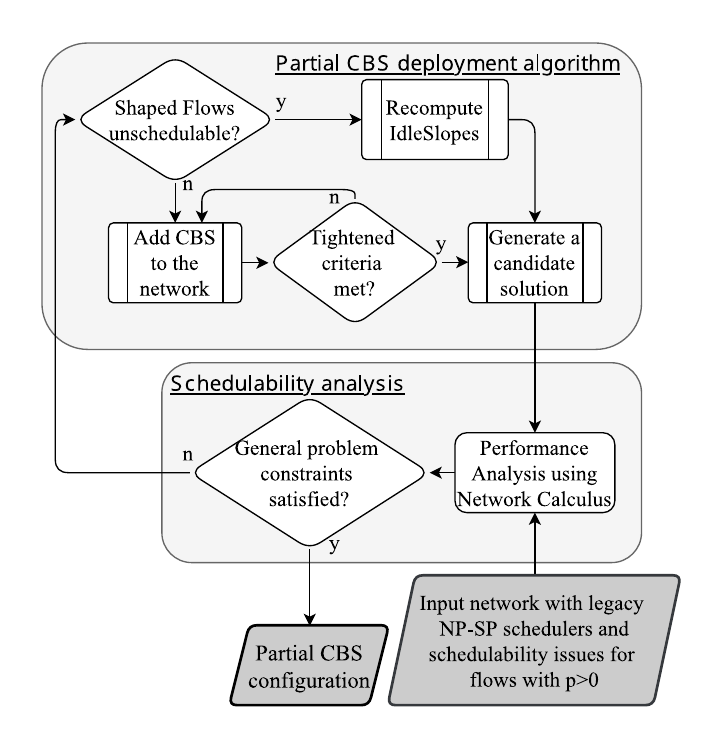}
\caption{Cost-efficient CBS deployment framework}
\label{fig:framework_overview}
\end{figure}
Our framework is based on two main steps:

(i) First, we verify the schedulability of the input configuration and identify the non schedulable flows (if any) based on Network Calculus (NC) \cite{le_boudec_network_2001}. The high modularity of such a theory and its applicability to networks with cyclic dependencies \cite{Thomas} make it particularly efficient for time-sensitive networks. In addition, the most relevant works concerning performance analysis of CBS mechanisms are based on NC \cite{de_azua_complete_2014,mohammadpour_end--end_2018,zhao_improving_2019}. This step is conducted in an iterative manner each time the network configuration is updated with new CBS devices until reaching a schedulable network configuration (if any). This step will be detailed in Section IV;

(ii) Second, we need to conduct the partial CBS deployment algorithm to place CBS selectively within some devices improving schedulability where it is most needed and avoiding the high cost of full CBS deployment. This step is based on the formulation of an optimization problem and its solving using a heuristic approach to cluster CBS within a minimum number of devices and configure their assigned bandwidths. It is worth noting that this step distinguishes two cases: (i)  when there are non-schedulable shaped priorities, it starts by updating the assigned bandwidths of concerned CBS to improve the schedulability of the shaped flows; (ii) when there are only non-shaped priorities with schedulability issues, it generates a CBS placement and configuration of assigned bandwidths. At the end of this step, the schedulability verification module is invoked to verify the performance of the generated CBS configuration. This step will be detailed in Sec. V and VI. 

As shown in Fig.\ref{fig:framework_overview}, there is a feedback loop between CBS deployment and schedulability analysis until meeting all the timing constraints (if a solution exists). The generated configuration with the CBS placement is then considered feasible and the process is stopped.

\section{Schedulability Verification Method}
\label{sec:schedulability}
We present in this section the used schedulability verification method based on Network Calculus (NC) \cite{le_boudec_network_2001}. We first present the Network Calculus framework and the main results for CBS mechanism. Then, we define the considered schedulability condition. Finally, we explain the computation of the upper bounds on end-to-end delays.

\subsection{Network Calculus Framework}
Network Calculus is a mathematical framework formally presented in \cite{le_boudec_network_2001}. It involves the concept of Arrival Curves ($\alpha(t)$) and Service Curves ($\beta(t)$) to describe the bounds on the amount of input data into a network element and the service offered by said element, respectively.  

In particular, for our network, we characterize a flow $f$ Arrival Curve by the leaky-bucket form, $\lambda_{b,r}(t)$:
\begin{equation}
    \alpha^f(t) = r^f\cdot t + b^f
\end{equation}

Moreover, we characterize the Service Curve offered by an output port $op$ to a traffic priority $p$ as a Rate-Latency curve form, $\beta_{R,T}(t)$:
\begin{equation}
\beta_{op}^p(t) = R_{op}^p[t-T_{op}^p]^+
\end{equation}

The delay and backlog bounds at a node are determined by the maximum horizontal and vertical distances between these curves. The calculus of these bounds is greatly simplified in the case of leaky bucket arrival curve $\lambda_{b,r}(t)$ and the Rate-Latency service curve $\beta_{R,T}(t)$. In this case, the delay is bounded by $\frac{b}{R} + T$ and the backlog bound is $ b + r*T$.
When considering the legacy NP-SP scheduler, each output port offers a rate-latency service curve to a given traffic priority computed through Cor. \ref{cor:residual-service-curve}.
\begin{Corollary}(Left-over service curve - NP-SP Multiplexing)\cite{bouillard_deterministic_2018}
\label{cor:residual-service-curve}
Consider a system with the strict service curve $\beta$ and $m$ flows crossing it, $f_1$,$f_2$,..,$f_m$. The maximum packet length of $f_i$ is $L^{i,max}$ and $f_i$ is $\alpha_i$-constrained. The flows are scheduled by the NP-SP policy, where priority of $f_i > $  priority of $f_j \Leftrightarrow i < j$. For each $i \in \{2,..,m \}$, the strict service curve of $f_i$ is given by\footnote{$g_{\uparrow}(t) = \max\{0,\sup_{0 \leq s \leq t} g(s)\}$}:

$$ (\beta - \sum_{j <i} \alpha_j - \max_{k \geq i} L^{k,max})_{\uparrow}$$		
\end{Corollary}

If the output port $op$ implements CBS for a given priority $p$ on top of NP-SP, then the provided service Curve is defined as \cite{mohammadpour_end--end_2018, zhao_improving_2019}:
\begin{equation}\label{eq: CBS Service Curve}
      \beta_{op}^p(t) = I_{op}^p\Big[t-\frac{cred^{p,max}_{op}}{I_{op}^p}\Big]^+
\end{equation}
Where $cred^{p,max}_{op}$ corresponds to the credit upper bound for the traffic-priority $p$ within the output port $op$:
\begin{equation}
cred_{op}^{p,max} = I_{op}^p\cdot \frac{\sum_{k=0}^{p-1}cred_{op}^{k,min}-L^{>p,max}}{\sum_{k=0}^{p-1}I_{op}^k-C}
\end{equation}

Here $cred^{p,min}_{op}$ corresponds to the credit lower bound for the traffic-priority $p$ within the output port $op$:
\begin{equation}
cred_{op}^{p,min} = (I_{op}^p-C)\cdot \frac{L^{p}}{C}
\end{equation}

\subsection{Sufficient Schedulability Condition}
\label{scheduTest}
To infer the schedulability analysis of our network configuration with partial CBS deployment, we define a sufficient schedulability condition, which consists in verifying that the end-to-end delay bound of each traffic flow computed with NC is lower than its deadline. 
	
The end-to-end delay bound expression of a flow $f$ with priority $p$, along its path $\phi^f$ is as follows:
\begin{equation}\label{eq: TFA}
    d^f = d_{src}^p + d_{prop}+\sum_{\forall op\in \phi^{f}\setminus src} d_{op}^p
\end{equation}

With $d_{src}^p$ the delay bound within the end-system ($src$) to transmit the aggregate traffic with priority $p$ and $d_{prop}$ the propagation delay along the path, which is generally negligible. The last delay $d_{op}^p$ represents the delay bound within each output port $op$ along the flow path and it is mainly due to the multiplexing delay of CBS and/or NP-SP scheduler. To enable the computation of upper bounds on these delays, we need a computation algorithm taking into account cyclic dependencies.

\subsection{Computing End-to-end Delay bounds}
\label{FP-TFA++}
Network Calculus-based algorithms have been developed to apply its principles to compute delay bounds at a network scale. The one used in this paper is the extended version of Total Flow Analysis (TFA++)\cite{mifdaoui_beyond_2017} taking into account cyclic dependencies in addition to line shaping effect, proposed in \cite{Thomas} and named Fixed-Point TFA++ (FP-TFA++). 
TFA++ is based on a three-step process: (i) Summing the arrival curves of individual flows of the same class that share the same input at each crossed output port; (ii) computing the delay bound, which considers the aggregate arrival curve of the flows and the global service curve guaranteed to a specific class within a given output port; (iii) determining the output arrival curve of each flow taking into account the line shaping effect using Equation \ref{eq: output arrival} for any flow $f$ with priority $p$ traversing $op$ with transmission capacity $C_{op}$:
\begin{equation}\label{eq: output arrival}
    \alpha^{*,p}_{op}(t) = min\left(C_{op}\cdot t,\alpha_{op}^p(t+d_{op}^p)\right)
\end{equation}

Where $d_{op}^p$ corresponds to the delay bound computed in step (ii).

For a network with cyclic dependencies, the main idea is to virtually perform cuts in the network graph so as to make it acyclic. Cuts can be selected freely as long as the remaining network is feed-forward. Afterwards, FP-TFA++ iterates the TFA++ algorithm on the remaining feed-forward network starting with an initial value of flow bursts until reaching a fixed point on the bursts. When this point is reached, the delay bounds can be computed as in step (ii) of TFA++ algorithm. This step of our framework is based on the computation tool WoPANets \cite{mifdaoui_wopanets_2010}.

\section{Partial CBS Deployment as an Optimization Problem}
\label{sec:deployment}
\subsection{Problem Formulation}\label{sec: Problem Definition}
The problem we aim to solve is a placement problem. Specifically, we have to select where to place CBS based on network context, i.e. flow attributes, flow end-to-end delay, flow path, already placed CBS, etc. Once placed, we also have to set an IdleSlope for each placed CBS so the purpose of guaranteeing schedulability is achieved, in addition to reducing the number of TSN devices throughout the network.

To master the clustering of CBS within the minimum of devices, we introduce the variable $TSN_{dev}$ which represents if a device implements CBS functionality. This is a binary variable that takes 1 if at least a CBS is implemented in $dev$ and 0 if none are implemented. 

\begin{equation}\label{eq: objective equation}
    TSN_{dev} = \begin{cases}
    1  & \text{if } \sum_{\forall op,p \mid op \in OP(dev)}CBS_{op}^{p}\geq 1\\
    0  & \text{if } \sum_{\forall op,p  \mid op \in OP(dev)}CBS_{op}^{p}< 1
  \end{cases}
\end{equation}

The optimization problem that we aim to solve is formulated as follows:
\paragraph*{Objective}
\begin{align}
    \text{minimize}\sum_{\forall dev \in DEV} TSN_{dev}
\end{align}

\paragraph*{Subject to}
\begin{align}
    &\forall dev \in DEV, \forall op \in OP(dev),\nonumber \\
    &\forall p \in P, \forall f \in \mathcal{F}, \nonumber\\
    & \sum_{\forall p \mid CBS_{op}^{p}=1} I_{op}^{p} \leq 0.75\cdot C  \tag{9a} \label{eq:9a}\\
    & \forall p \mid CBS_{op}^{p}=1, I_{op}^{p} \geq r^p  \tag{9b} \label{eq:9b}\\
    & CBS_{op}^{p} \geq CBS_{op}^{p+1}  \tag{9c} \label{eq:9c}\\
    & d^f \leq D^f  \tag{9d} \label{eq:9d} \\
    & CBS_{op}^{p} \in \{0,1\} \tag{9e} \label{eq:9e}
\end{align}
where,
\begin{itemize}
\item The objective is to minimize the number of TSN devices deployed in the network;
\item Constraint (9a) guarantees that the sum of all IdleSlopes of CBS implemented in an output port is lower than $75\%$ of the output's bandwidth, as defined in the standard \cite{8021QStd};
\item Constraint (9b) guarantees that the assigned IdleSlope of a CBS for the shaped priority respects the stability condition of the priority level;
\item Constraint (9c) guarantees that any given priority where CBS is implemented have a higher priority than any traffic classes not implementing it, as intended by the standard \cite{8021QStd} in \textit{note 5} of section; \textit{8.6.8.2 Credit-Based Shaper Algorithm};
\item Constraint (9d) guarantees that the schedulability condition is verified for all the flows;
\item Constraints (9e) guarantees that $CBS_{op}^{p}$ is a binary variable.
\end{itemize}

All the constraints have to be guaranteed for any device, any output port, any priority and any flow in the network.

As presented, this problem corresponds to a Non-Linear Problem (NLP) that is commonly known to be hard to solve. The main complexity is due to Equation $(9d)$ relying on an oracle's verification, implementing FP-TFA++ algorithm and its recursive nature, which is time consuming. To overcome these issues, we address the problem through relaxing it using linearization and constraints tightening techniques that will be detailed in the next section. 

\subsection{Problem Linearization \& Constraint tightening}
\label{sec: problem linearization and constraint tightening}
We start by linearizing the objective function. This is enabled through adding linear constraints to simplify the definition of $TSN_{dev}$ binary variable in Eq. (\ref{eq: objective equation}) as follows:
\begin{align}
\label{Eq:linear-objective}
    & \forall dev \in DEV,\forall (op,p) \,|\, op \in OP(dev), \nonumber \\
    & TSN_{dev} \geq CBS^{p}_{op} \nonumber \\
    & CBS_{op}^{p} \in \{0,1\} \nonumber \\
    & TSN_{dev} \in \{0,1\}
\end{align}

Afterwards, during the CBS deployment step of our proposed framework (see Fig. \ref{fig:framework_overview}), we tighten the schedulability constraint in Eq. (9d). First, we remove the recursive computation of the delay bound due to FP-TFA++ algorithm through ignoring the burst propagation phenomena along the flow path. Then, we focus only on guaranteeing the shaped flows schedulability when placing and configuring CBS, since the verification of all the schedulability constraints will be conducted during the schedulability verification step. Consequently, during the CBS deployment step, Eq (9d) is verified only for:\\
\[ \forall p \in P, \forall f \in \mathcal{F}(p)  \,|\, \exists op, CBS_{op}^p =1\]

Assuming the input network satisfies the initial condition where all flows of priority \( p=0 \) are schedulable, our goal is to maintain their schedulability while improving the performance of medium and lower priority flows (\( p>0 \)). Since implementing CBS reduces the burstiness of shaped flows and consequently improves the schedulability of non-shaped flows, assigning the minimum IdleSlope that guarantees the schedulability of shaped flows will offer the best improvement for non-shaped flows. Hence, the idea is to define a minimum IdleSlope for each placed CBS while guaranteeing the schedulability of shaped flows.

To address this, we start by defining a local deadline for each output port $op$ and each priority level $p$ when $CBS_{op}^p =1$.

\textbf{Definition:} For each output port $ op $ and each priority level $ p $ where $ CBS_{op}^p = 1 $, there exists a local deadline $ D_{op}^p $ defined as:
\begin{align}\label{eq: local deadline}
    &\forall (op,p)  \,|\, CBS_{op}^p =1,\nonumber \\
    & D_{op}^p = \min_{f \in \mathcal{F}(op,p)} \left[(D_f- d_{src}^p-d_{prop}) \cdot \frac{r_{op}^{p}}{\sum_{\forall op' \in \phi_f \backslash src} r_{op'}^p}\right]
\end{align}
This local deadline is computed for each output port $op$ and priority $p$ when considering the most constrained shaped flow $\in \mathcal{F}(op,p)$ for which: (i)the end-to-end delay bound in (Eq. \ref{eq: TFA}) is lower than its assigned deadline; (ii) its deadline is distributed proportionally to the utilization rate of crossed output ports.
\begin{proof}
Starting from Equation 9d, by substituting Equation \ref{eq: TFA} into Equation 9d, we obtain for any $f \in \mathcal{F}(op,p)$
    \[
    \sum_{\forall op \in \phi^f\setminus src} d_{op}^{p} \leq D^f - d_{src}^{p} - d_{prop}
    \]

    There exists a local deadline \( D_{op}^{p} \) such that:

    \begin{align}
\label{eq: local-deadline-bis}
\sum_{\forall op \in \phi^f\setminus src} d_{op}^{p} & \leq \sum_{\forall op \in \phi^f\setminus src} D_{op}^{p} \nonumber \\
     & \leq D^f - d_{src}^{p} - d_{prop} \nonumber \\
     & \leq \min_{f \in \mathcal{F}(op,p)} \left[D^f - d_{src}^{p} - d_{prop}\right]
\end{align}

    Consider \( D_{op}^{p} \) as:
     \[
    D_{op}^p = \min_{f \in \mathcal{F}(op,p)} \left[(D_f - d_{src}^p - d_{prop}) \cdot \frac{r_{op}^{p}}{\sum_{\forall op' \in \phi_f \backslash src} r_{op'}^p}\right]
    \]

    Summing \( D_{op}^{p} \) over all output ports in the path of any flow $f \in \mathcal{F}(op,p)$ (excluding the source) verifies Eq. \ref{eq: local-deadline-bis}.
    This completes the proof.
\end{proof}

Afterwards, we ignore the burst propagation phenomena when computing the delay bound of a shaped flow within a crossed output port implementing CBS. This is to discard the recursive computation of the delay bound due to FP-TFA++ Algorithm. Here, the local schedulability condition becomes:
\begin{align}\label{eq: idleSlope computation}
&\forall p,  \,|\, \exists op, CBS_{op}^p =1,\nonumber \\
&  \frac{\sum_{f \in \mathcal{F}(op,p) } b^{f}}{I_{op}^{p}} + T_{op}^{p} &\leq D_{op}^p \cdot \text{margin}
\end{align}

Where variable $margin \in [0,1]$ is introduced to compensate the optimism of the computed delay bound when ignoring the burst propagation phenomena with reference to the delay bound commonly computed using NC.

From Eq. (\ref{eq: idleSlope computation}), we derive the minimum IdleSlope assigned to a CBS for a priority $p$ within the output port $op$:
\begin{equation}\label{eq: minimum idleslope}
    I^{p}_{op} \geq \frac{\sum_{f \in \mathcal{F}(op,p) } b^{f}}{D^{p}_{op} \cdot \text{margin} - T_{op}^{p}} =I^{p,min}_{op}
\end{equation}

where, $T_{op}^{p}$ corresponds to the latency factor from $\beta^{p}_{op}(t)$ in Eq. (\ref{eq: CBS Service Curve}):
\begin{equation}\label{eq: latency of the service curve}
    T_{op}^{p} = \frac{\sum_{k=0}^{p-1}cred_{op}^{k,min}-L^{>p,max}}{\sum_{k=0}^{p-1}I_{op}^{k}-C}
\end{equation}
Eq. (\ref{eq: latency of the service curve}) shows that to calculate a minimum IdleSlope using Eq. (\ref{eq: minimum idleslope}), the IdleSlopes of higher-priority CBS within the same $op$ must be determined first. By adhering to Constraint (9c), we ensure that CBS are placed sequentially (starting with the highest priority), which simplifies the computation of Eq. (\ref{eq: minimum idleslope}), as there are no unknown variables upon computation.

The final simplified optimization problem is as follows:

\paragraph*{Objective}
\begin{equation}
    {\text{minimize}} \quad \sum_{dev \in DEV} TSN_{dev}
\end{equation}

\paragraph*{Constraints}
\begin{align}
& \sum_{\forall p \mid CBS_{op}^{p}=1} I_{op}^{p} \leq 0.75\cdot C , \quad \forall dev, \forall (op,p) \nonumber\\
&  I_{op}^{p} \geq r^p , \quad \forall (op,p) \mid CBS_{op}^{p}=1 \nonumber\\
& CBS_{op}^{p} \geq CBS_{op}^{p+1} ,\quad \forall dev, \forall (op,p) \nonumber\\
& TSN_{dev} \geq CBS^{p}_{op} , \quad \forall dev, \forall (op,p) \nonumber \\
& I^{p}_{op} \geq \frac{\sum_{f \in \mathcal{F}(op,p) } b^{f}}{D^{p}_{op} \cdot \text{margin} - T_{op}^{p}}, \quad \forall dev, \forall (op,p) \nonumber\\
& CBS_{op}^{p} \in \{0,1\}, TSN_{dev} \in \{0,1\}, \quad \forall dev,\forall (op,p) \nonumber
\end{align}

This formulation provides an Integer Linear Programming problem (ILP). However, in time-sensitive networks with $>10$ devices and $>1000$ flows, the number of constraints to produce a feasible solution can exceed 10000, this is computationally inefficient for ILP solvers \cite{noauthor_tolerances_nodate}. Therefore, we choose to address this problem with a custom heuristic approach that is detailed in the next section.

\section{Cost-aware Heuristic Approach for CBS deployment}
\subsection{Heuristic Algorithm}
Our heuristic approach to place CBS is illustrated through Algorithm \ref{alg:partial_cbs_function} that operates under two primary scenarios. The first addresses cases where shaped flows are not schedulable (referred to as \texttt{unschedShaped} in the algorithm). In such cases, a reconfiguration is triggered (Lines 2--10) to adjust overestimated IdleSlope values from a prior execution. Initially, the margin $m$ is updated (Line 4) using the constant $\rho$ as follows:  
\begin{equation}
m = m - \rho.
\end{equation}  
After updating the margin, its value is compared to the minimum theoretical margin in Lines 5--7. If the new margin is less than or equal to this limit, the algorithm concludes that no feasible solution can be found and terminates (Line 6). The minimum theoretical margin, is calculated as follows:  
\begin{equation}
\max\left(\frac{d_{op}^{p}}{D^{p}_{op}}\right).
\end{equation}  

where $D^{p}_{op}$ is computed based on Eq. (\ref{eq: local deadline}).

In Line 8, the \texttt{recomputeIdleSlopes} function recalculates the IdleSlope values for previously placed CBS using the updated margin. The resulting network configuration is then returned for schedulability verification (Line 9).  

In contrast, if all shaped flows are schedulable under the current configuration, but non shaped flows (denoted in the algorithm as \texttt{unschedNonShaped}) remain unschedulable, the algorithm proceeds to the CBS placement phase (Lines 11--35). In this phase, it attempts to deploy CBS in a way that benefits all flows in the network. This process builds upon the existing configuration by keeping previously placed CBS unchanged while adding new ones.  

\begin{algorithm}
\caption{Partial CBS Algorithm}
\label{alg:partial_cbs_function}
\begin{algorithmic}[1]
    \Function{partialCBSDeployment}{unschedShaped, unschedNonShaped,m,$\rho$}
        \If{unschedShaped $\neq \emptyset$}
            \State // Perform Reconfiguration
            \State m $\gets$ m$-\rho$
            \If{m $\leq m_{min}$}
                \State \textbf{return} null \Comment{No solution could be found}
            \EndIf
            \State recomputeIdleSlopes(m)
            \State \textbf{return} \textit{network configuration}
        \EndIf
        \State devsToExclude $\gets \emptyset$
        \State shapedFlows $\gets \emptyset$
        \While{unschedNonShaped $\neq \emptyset$}
            \State // Place CBS
            \State foi $\gets$ mostConstrainedFlow(unschedNonShaped)
            \State \textbf{device selection:}
            \State dev $\gets$ getDev(foi, devsToExclude)
            \State ops $\gets$ getOrderedOps(dev)
            \For{\textbf{all} op $\in$ ops}
                \State p $\gets$ highestNonShapedPriority(op)
                \State $I_{op}^{p} \gets$ computeIdleSlope(op,$f_{verif}$)
                \If{$\sum_{\forall p}I_{op}^{p} \leq 0.75\cdot C$}
                    \State $CBS_{op}^{p} \gets 1$
                    \State shapedFlows.add(\(\{ f \mid f \in \mathcal{F}(op,p)) \}\))
                \ElsIf{$op\in \phi_{foi}$}
                    \State devsToExclude.add(dev)
                    \State \textbf{go to} \textit{device selection}
                \EndIf
            \EndFor
            \State DA $\gets$ getDirectlyAffectedFlows()
            \State IA $\gets$ getIndirectlyAffectedFlows()
            \State unschedNonShaped.remove($DA \cup IA$)
        \EndWhile
        \State \textbf{return} \textit{network configuration}
    \EndFunction
\end{algorithmic}
\end{algorithm}
The process begins with initializing the \texttt{devsToExclude} list (Line 11), which tracks devices excluded from CBS deployment when these do not fit the constraints from our problem. In Line 12 initializes \texttt{shapedFlows} which tracks the shaped flows issued from CBS placement during the present execution of the algorithm. It then enters a \texttt{while} loop (Line 13) that continues until all non-shaped unscheduled flows in \texttt{unschedNonShaped} are addressed. Initially, \texttt{unschedNonShaped} contains all non-shaped unscheduled flows identified by the verification tool.  

The first step involves identifying the most critical device for CBS placement. To this effect, the algorithm selects the most constrained unscheduled and non-shaped flow in the network, called \texttt{foi} (Line 15). This flow must have the highest priority level in \texttt{unschedNonShaped}. Specifically, this flow is obtained as:
\begin{equation}
foi = \arg\max_{f \in \texttt{unschedNonShaped}|p=p^{max}_{unschedNonShaped}} (d^f - D^f).
\end{equation}
Where $p^{max}_{unschedNonShaped}$ corresponds to the highest priority level flows inside \texttt{unschedNonShaped}.

Line 16 provides a reference point for re-selecting a device if needed. Line 17 identifies the device where \texttt{foi} first encounters higher-priority flows that are unshaped at that specific device, considering only devices where CBS deployment is still feasible. CBS deployment is considered possible as long as there are unshaped priority levels available for which only schedulable flows pass through. The algorithm then sorts the output ports on the selected device based on the most constrained flow per port. This is done on Line 18 and the result is placed on set $ops$.

The algorithm then iterates over these ordered ports (Lines 19--29) to place CBS. Line 21 computes the IdleSlope for ($op,p$)  using Eq. \ref{eq: minimum idleslope} by using the corresponding $f_{verif}$ from Eq. \ref{eq: local deadline} (the flow with the most constrained deadline). If this value satisfies the constraint (9a) (Line 22), a CBS is deployed at ($op,p$) in Line 23 and the \texttt{shapedFlows} list is updated with flows that are now shaped. If not, and it is the first $op$ to be explored on $dev$ (Line 25), the current device is excluded from further consideration (Line 26), and the algorithm returns to Line 16 to select a new device (Line 27).  In case the computed IdleSlope does not comply with constraint (9a), the $op$ is ignored for CBS placement and the remaining $op \in ops$ are explored. 

Once all ports on the device have been explored, the algorithm identifies two sets of flows: \textit
{DA (Directly Affected Flows)} for Unscheduled flows in \texttt{unschedNonShaped} traversing the current device (Line 30). DA is formally computed as:
\begin{equation}
    DA = \{f \mid f \in (unschedNonShaped\cap \mathcal{F}(dev))\},
\end{equation}
and \textit{IA (Indirectly Affected Flows)} for Unscheduled flows in \texttt{unschedNonShaped} that share paths with flows in DA or \texttt{shapedFlows}, but do not traverse the current device (Line 31). Formally, it is computed as:
\begin{align}
    &IA = \{  f \in \text{\texttt{unschedNonShaped}} \mid \nonumber\\ 
    &f \notin \mathcal{F}(\text{dev}) \land (\exists f' \in DA \cup \text{\texttt{shapedFlows}} \mid 
    \phi^f \cap \phi^{f'} \neq \emptyset )\} \nonumber\\  
\end{align}

Both sets are removed from \texttt{unschedNonShaped} (Line 32), since they may be impacted by the placed CBS. If \texttt{unschedNonShaped} becomes empty, the algorithm returns the network containing CBS placement and configuration as a potential solution (Line 34) to verify its schedulability.  

\subsection{Illustrative example}\label{sec:performance}
To illustrate how Algorithm \ref{alg:partial_cbs_function} works, we use the network on Figure \ref{fig:RepresentativeAvionicNetworkArchitecture}. It is traversed by flows from Table \ref{tab:flow_characteristics} which also contains the flow's characteristics. It has three priorities, for which the third is considered best effort. Table \ref{tab:input} gives the variables of interest for the algorithm's execution as well as the results computed after a valid solution has been found.

\subsection*{Step-by-Step Execution}x`
The initial set of unscheduled non-shaped flows is \(\{f_2, f_3, f_4\}\) and \texttt{margin} is set to $1$. Below, we outline the line-by-line results of the algorithm's execution:

\begin{enumerate}
    \item \textbf{Line 2}: unschedShaped = $\emptyset$
    \item \textbf{Line 13}: foi $\gets f_2$ .
    \item \textbf{Line 15}: dev $\gets SW0$ 
    \item \textbf{Line 16}: ops $\gets (SW0_2, SW0_1)$ 
    \item \textbf{Lines 19--29 (First iteration of the \texttt{for} loop)}:
    \begin{itemize}
        \item op $\gets SW0_2$
        \item p $\gets$ 0
        \item $I_{SW0_2}^{0} = 44460912$bps
        \begin{itemize}
            \item Where $f_{verif}=f_0$
        \end{itemize}
    \end{itemize}
    \item \textbf{Lines 19--29 (Second iteration of the \texttt{for} loop)}:
    \begin{itemize}
        \item op $\gets$\(SW0_1\)
        \item p$\gets$ 0
        \item $I_{SW0_1}^{0} = 34416827$ bps
        \begin{itemize}
            \item With $f_{verif}=f_1$
        \end{itemize}
    \end{itemize}
    \item \textbf{Lines 30--32}: 
    \begin{itemize}
        \item DA $\gets$ \(\{f_2, f_3\}\)
        \item IA $\gets$ \(\{f_4\}\)
        \item \texttt{unschedNonShaped.remove}($\{f_2, f_3,f_4\}$)
    \end{itemize}
    \item \textbf{Line 13}: The set \texttt{unschedNonShaped} is now empty, and the algorithm terminates.
\end{enumerate}

The schedulability of the proposed solution is verified and it is considered as a valid solution. The computed delay bounds under partial CBS are detailed in Table \ref{tab:input}. Consequently, the framework returns the enriched network configuration, including the CBS setup computed by the algorithm. This configuration successfully schedules all previously unscheduled flows while clustering all CBS onto a single device.
\begin{table}[t]
    \centering
    \caption{Flow Characteristics}
    \small
    \begin{tabular}{c|c|c|c|c|c}
        \textbf{Flow} & \textbf{P} & \makecell{\textbf{b}\\(kbit)} & \makecell{\textbf{Period}\\(ms)} & \makecell{\textbf{r}\\(Mbps)} & \textbf{Path} \\ \hline \hline
        $f_0$ & 0 & 14.4 & 1 & 14.4 & $ES_0 \to ES_2$ \\ \hline
        $f_1$ & 0 & 14.4 & 1 & 14.4 & $ES_2 \to ES_1$ \\ \hline
        $f_2$ & 1 & 0.96 & 1 & 0.96 & $ES_1 \to ES_2$ \\ \hline
        $f_3$ & 1 & 0.96 & 1 & 0.96 & $ES_2 \to ES_1$ \\ \hline
        $f_4$ & 1 & 0.96 & 1 & 0.96 & $ES_3 \to ES_2$ \\ \hline
        $f_5$ & 2 & 123.4 & $\infty$ & 1.234 & $ES_3 \to ES_2$ \\ \hline
    \end{tabular}
    \label{tab:flow_characteristics}
\end{table}
\begin{table}[t]
    \centering
    \begin{tabular}{c|p{1.5cm}|p{1.6cm}|p{1.7cm}|p{1.3cm}}
        \text{Flow} & \text{src Delay} \newline \text{Bound} ($\mu s$) & \text{Delay bound} \newline \text{(NP-SP)}($\mu s$) & \text{Delay bound} \newline\text{(Partial CBS)} ($\mu s$) & \text{Deadline} \newline ($\mu s$) \\ \hline \hline
        $f_0$ & 216 & 314 & 479 & 1000 \\ \hline
        $f_1$&  285 & 213 & 378 & 1000 \\ \hline
        $f_2$ & - & 543 & 396 & 535 \\ \hline
        $f_3$& - & 560 & 419 & 555 \\ \hline
        $f_4$& - & 473 & 472 & 472 \\ \hline
        $f_5$& - & - & - \\ \hline
    \end{tabular}
    \caption{Variables of interest to the CBS deployment algorithm from the verification tool}
    \label{tab:input}
\end{table}
\subsection{Complexity Analysis}
\begin{theorem}\label{th:complexity}
The computational complexity of the \textit{partialCBSDeployment} algorithm is 
\[
O\left(|\mathcal{F}|^3 \cdot |DEV| +|\mathcal{F}|^2\cdot |DEV|^2\cdot log(|DEV|)\right)
\]

where $|\mathcal{F}|$ represents the number of flows traversing the network and $|DEV|$ is the number of devices in the network.
\end{theorem}

\begin{proof}
To compute the complexity of the \texttt{partialCBSDeployment} algorithm, we analyze the cost of each step. If it is not specified, it means instructions can be performed in constant time.
\begin{enumerate}
    \item Initialization:
   \begin{itemize}[leftmargin=*]
       \item \texttt{recomputeIdleSlopes}: Has a complexity of $O(|DEV| \cdot OP_{dev})$, where $OP_{dev}$ represents the number of output ports for a device. In the worst case, $dev$ will be connected to every other device in the network, thus $OP_{dev}=|DEV|-1$, leading to a simplified complexity of $O(|DEV|^2)$.
   \end{itemize}
   \item Outer \texttt{while} loop:
    \begin{itemize}
        \item The \texttt{while} loop iterates until \texttt{unschedNonShaped} becomes empty. In the worst case, \texttt{unschedNonShaped} contains $|\mathcal{F}(p>0)|$ flows. Since one flow is removed per iteration, the loop executes $|\mathcal{F}(p>0)| + (|\mathcal{F}(p>0)| - 1) + \dots + 1 = O(|\mathcal{F}(p>0)|^2)$ times.
    \end{itemize}
    \item Inside the \texttt{while} loop:
   \begin{itemize}
       \item \texttt{mostConstrainedFlow}: $O(|unschedNonShaped|)$, which is $O(|\mathcal{F}(p>0)|)$ in the worst case.
       \item \texttt{getDev}: $O(|DEV|)$.
       \item \texttt{getOrderedOps}: $O(OP_{dev} \log(OP_{dev}))$ for sorting. Therefore, in the worst case $O\left(|DEV|\cdot log(|DEV|)\right)$.
       \item \texttt{for} loop over $|ops|$: Each device has $OP_{dev}$ output ports, so the loop runs $O(OP_{dev})$ times. In the worst case, the loop runs $|DEV|$ times.
       \item Inside the \texttt{for} loop:
           \begin{itemize}
               \item Re-selection of the device (\texttt{go to}): Happens at most once per \texttt{for} loop iteration.
           \end{itemize}
       \item Computing \texttt{DA}: $O(|unschedNonShaped| \cdot |\phi^f|)$, where $|\phi^f| \leq |DEV|$ in the worst case. Thus, $O(|\mathcal{F}(p>0)| \cdot |DEV|)$.
       \item Computing \texttt{IA}: $O(|unschedNonShaped| \cdot |\phi^f| \cdot (1 + |\texttt{shapedFlows}|))$. With $|\phi^f| \leq |DEV|$, this becomes $O(|\mathcal{F}(p>0)| \cdot |DEV| \cdot |shapedFlows|)$. Considering $|shapedFlows|=|\mathcal{F}|-|\mathcal{F}(p>0)|$ in the worst case, we denote this substraction as $|\mathcal{F}(0)$. Thus, the complexity becomes $O\left(|\mathcal{F}(p>0)| \cdot |DEV| \cdot |\mathcal{F}(0)|\right)$.
   \end{itemize}
   \item Total Complexity:

    Aggregating these factors and removing those that are from the expression, we obtain:
    
\end{enumerate}
\begin{align*}
O\big(&|\mathcal{F}(p > 0)|^3 \cdot |DEV| \\
     &+ |\mathcal{F}(p > 0)|^2 \cdot |DEV|^2 \cdot \log(|DEV|) \\
     &+ |\mathcal{F}(p > 0)|^2 \cdot |DEV| \cdot |F_{p = 0}|\big)
\end{align*}
\item Since $|\mathcal{F}(p>0)|$ and $|\mathcal{F}(0)|$ can both be bounded by $|\mathcal{F}|$, the expression is simplified to:
\[
O\left(|\mathcal{F}|^3 \cdot |DEV| +|\mathcal{F}|^2\cdot |DEV|^2\cdot log(|DEV|)\right)
\]
\end{proof}

From the result of Theorem \ref{th:complexity}, we observe that the complexity of running the algorithm depends on the number of devices and flows traversing it.

\section{Industrial Use Case Evaluation}
\label{sec:industrial}
Figure \ref{fig: Adapted Industrial Automotive Topology} illustrates the industrial use case adapted from \cite{cuenot_2021}, representing an automotive network architecture. Table \ref{tab:flow_characteristics_indus} details the flow characteristics traversing the network. For this network, we analyze the number of TSN devices and deployed CBS, the Cumulative Distribution Function (CDF) of delays for high-priority shaped flows (Partial CBS vs. Full CBS), and the CDF of delays for medium-priority non-shaped flows (Partial CBS vs. NP-SP only). Here, the IdleSlopes for Full CBS deployment were computed by selecting the minimum IdleSlope value allowing for the schedulability of every shaped flow.
\begin{figure}[h]
	\centering
	\includegraphics[width=\columnwidth]{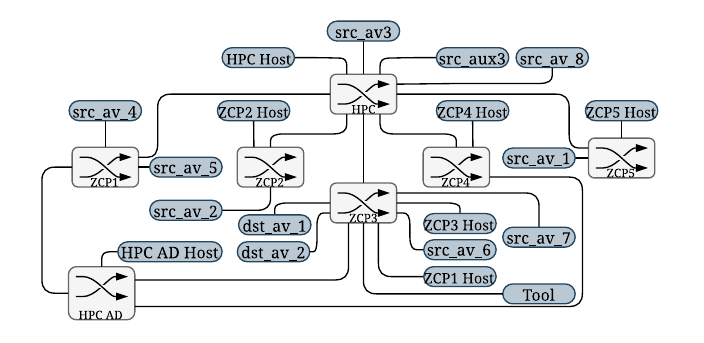}
	\caption{Adapted Industrial Automotive Topology from \cite{cuenot_2021}}
	\label{fig: Adapted Industrial Automotive Topology}
\end{figure}

\begin{table}[h]
    \caption{Adapted Flow Characteristics for the Industrial Automotive Topology \cite{cuenot_2021}}
    \label{tab:flow_characteristics_indus}
    \centering
    \small
    \begin{adjustbox}{max width=\columnwidth}
    \begin{tabular}{p{1cm} p{1cm} p{1cm} p{1.3cm} p{1.3cm} p{2cm} p{2cm}}
        \toprule
        \textbf{Name} & Priority &\textbf{rate (kB/s)} & \textbf{Deadline (ms)} & \textbf{Burst (B)} & \textbf{Sources} & \textbf{Destinations} \\
        \midrule
        AVB\_1 & 0 & 2224 & 1.25  & 2780 & src\_av\_1 & dst\_av\_\{1-2\} \\ \hline
        AVB\_2 & 0 & 472 & 10  & 15570 & src\_av\_\{3-7\} & HPC host \\ \hline
        AVB\_3 & 0 & 1307 & 33  & 43140 & src\_av\_8 & HPC host \\ \hline
        CAN\_1 & 1 & 613 & 0.075  & 46 & ZCP1 Host, HPC Host & ZCP\{1-5\}, HPC Host \\ \hline
        CAN\_2 & 1 &  1040 & 0.075  & 78 & ZCP1 Host, HPC Host & ZCP\{1-5\}, HPC Host \\ \hline
        CAN\_3 & 1 & 1146 & 0.075  & 86 & ZCP1 Host, HPC Host & ZCP\{1-5\}, HPC Host \\ \hline
        CAN\_4 & 1 & 560 & 0.075  & 42 & ZCP1 Host, HPC Host & ZCP\{1-5\}, HPC Host \\ \hline
        FILE & 2 & 63840 & -  & 12768000 & HPC Host & Tool \\ \hline
        \bottomrule
    \end{tabular}
    \end{adjustbox}
\end{table}

\begin{table}[h]
    \centering
    \caption{Deployed CBS for the Industrial Use-Case Network}
    \label{tab:indus}
    \begin{tabular}{c|c|c} 
        \textbf{Device} & \textbf{Output Port} & \textbf{IdleSlope} \\ \hline \hline
        HPC & HPC $\rightarrow$ HPC\_Host & 123 249 864 bps \\ \hline
        HPC & HPC $\rightarrow$ ZCP3 & 46 419 984 bps \\ \hline
        ZCP3 & ZCP3 $\rightarrow$ HPC & 123 281 022 bps \\ \hline
    \end{tabular}
\end{table}

\textbf{Impact on Network Costs}: Table \ref{tab:indus} presents the IdleSlopes computed by the framework and the corresponding output ports. These results reveal that only $2/7$ devices are TSN-enabled, representing a $70\%$ reduction compared to a Full CBS deployment. Furthermore, only $3/34$ CBS instances were deployed, resulting in a $91\%$ reduction.

\textbf{Impact on shaped flows performance}: Figure \ref{fig:CDF for delays of flows p=0} compares the CDFs of delays for shaped, high-priority flows under Partial CBS and Full CBS deployments. The horizontal deviation between the curves quantifies the improvement, with a maximum delay reduction of $24\%$ and a minimum of $4\%$. Notably, the Partial CBS CDF is consistently on top over the Full CBS CDF, demonstrating reduced blocking effects and improved delays for shaped flows.

\textbf{Impact on non-shaped flows schedulability}: Figure \ref{fig:CDF for delays of flows p=1} analyzes the CDFs of unshaped, medium-priority flows under Partial CBS and NP-SP. The plot also highlights the maximum deadline. Additionally, while some flows remain unschedulable under NP-SP, all flows meet their deadlines with Partial CBS deployment, where schedulability improves by $60\%$ under Partial CBS. This result shows the algorithm's effectiveness to reduce unshaped flows delays by leveraging CBS fairness enhancement.

\begin{figure}[h!] 
    \centering
    \includegraphics[width=0.9\linewidth]{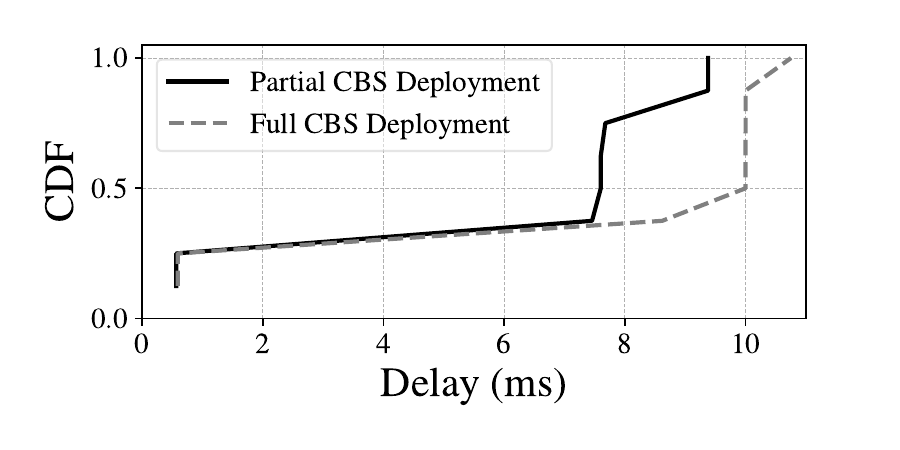}
    \vspace{-0.25cm}
    \caption{CDF for delay bounds of shaped flows, p=0}
     \label{fig:CDF for delays of flows p=0}
\end{figure}

\begin{figure}[h!] 
    \centering
    \hspace{-1cm}
    \includegraphics[width=\linewidth]{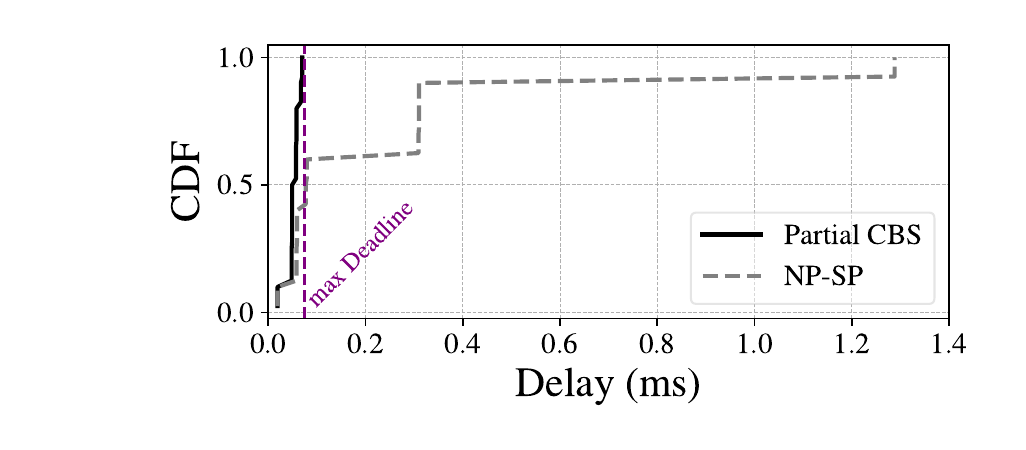}
    \vspace{-0.25cm}
    \caption{CDF for delay bounds of non-shaped flows, p=1}
     \label{fig:CDF for delays of flows p=1}
\end{figure}

Finally, we evaluate the framework's execution time. Processing the entire framework for this network takes $60$ ms. The placement algorithm required only a single execution to find a valid solution, taking $38$ ms, while the verification process consumed $22$ ms. These results show the efficiency of the algorithm to rapidly converge to a solution.

\section{Conclusions and Future Work}
\label{sec: conclusion}
This paper proposes a Partial CBS deployment framework to reduce TSN device usage and costs and validates its benefits on a representative automotive use case. The proposed approach in this paper effectively clusters CBS to key devices which effectively decreases deployment costs (up to 70\% of reduction). It also minimizes blocking effects on shaped flows with reference to Full CBS, while enhancing medium-priority flow schedulability (up to 60 \%) compared to legacy NP-SP solution.

Future work will perform a more in-depth performance evaluation of our approach by looking at factors such as efficiency and scalability through modifying network size and number of flows. We will also consider load variation and interference patterns to verify the performance of Partial CBS deployment. Extensions to the approach will include accommodating preset legacy devices in the optimization process. We will also consider CBS deployment in the context of Frame Replication and Elimination in the context of time-sensitive networks.

%
\IEEEpeerreviewmaketitle

\ifCLASSOPTIONcaptionsoff
  \newpage
\fi



%

\bibliographystyle{IEEEtran}
\vspace{-0.05in}
\bibliography{ref}

%
%
%
%
%
%
%
%
\end{document}